\documentclass[11pt]{article}

\makeatletter
\renewcommand*\@fnsymbol[1]{\the#1}
\makeatother

\usepackage{amssymb}
\usepackage{amsfonts}
\usepackage{amsmath}
\usepackage{amsthm}
\usepackage[english]{babel}
\usepackage{latexsym}
\usepackage{color}
\usepackage{enumerate}
\usepackage{nicefrac}
\usepackage{mathrsfs}
\usepackage{comment}
\usepackage[hmargin=2cm,vmargin=4cm]{geometry}
\usepackage[affil-it]{authblk}

\theoremstyle{plain}
\newtheorem{theorem}{Theorem}[section]
\newtheorem{lemma}[theorem]{Lemma}
\newtheorem{proposition}[theorem]{Proposition}
\newtheorem{corollary}[theorem]{Corollary}
\theoremstyle{definition}
\newtheorem{definition}[theorem]{Definition}
\theoremstyle{remark}

\newtheorem{remark}[theorem]{Remark}
\newtheorem{example}[theorem]{Example}

\numberwithin{equation}{section}

\newcommand{\ba}{\begin{array}{ll}}
\newcommand{\bal}{\begin{array}{ll}}
\newcommand{\ea}{\end{array}}

\newcommand{\E}{\mathbb{E}}

\newcommand{\probp}{\mathbb{P}}
\newcommand{\probq}{\mathbb{Q}}
\newcommand{\R}{\mathbb{R}}
\newcommand{\N}{\mathbb{N}}

\newcommand{\cF}{{\mathscr{F}}}

\newcommand{\cA}{\mathscr{A}}

\newcommand{\cone}{\mathop{\rm cone}\nolimits}

\newcommand{\e}{\varepsilon}

\newcommand{\VaR}{\mathop {\rm VaR}\nolimits}
\newcommand{\TVaR}{\mathop {\rm TVaR}\nolimits}
\newcommand{\SPAN}{\mathop {\rm SPAN}\nolimits}

\DeclareMathOperator*{\esssup}{ess\,sup}

\def\keywords{\vspace{.5em}
{\noindent\textbf{Keywords}:\,\relax%
}}

\def\JELclassification{\vspace{.5em}
{\noindent\textbf{JEL classification}:\,\relax%
}}

\def\MSCclassification{\vspace{.5em}
{\noindent\textbf{MSC}:\,\relax%
}}

\makeatletter
\def\@fnsymbol#1{\ensuremath{\ifcase#1\or *\or 1\or 2\or
   3\or 4\or 5\or 6\or 7\or 8\else\@ctrerr\fi}}
\makeatother

\begin{document}
%%%%%%%%%%%%%%%%%%%%%%%%%%%%%%%%%%%%%%%%%%%%%%%%%%%%%%%%%%%%%%%%%%%%%%%%%%%%%%%%%%%%%%%%%%%%%%%%%%%%%%%%%%%%%%%%%%%%%%%%%%%%%%%%%%%%%

\title{Capital adequacy tests and\\
limited liability of financial institutions\footnote{We would like to thank J. Staum for valuable comments on a first version of this paper. Partial support through the SNF project 51NF40-144611 ``Capital adequacy, valuation, and portfolio selection for insurance companies'' is gratefully acknowledged.}}

\author{\sc{Pablo Koch-Medina}\thanks{Email: \texttt{pablo.koch@bf.uzh.ch}}\,, \sc{Santiago Moreno-Bromberg}\thanks{Email: \texttt{santiago.moreno@bf.uzh.ch}}}
\affil{Department of Banking and Finance, University of Zurich, Switzerland}

\author{\sc{Cosimo Munari}\,\thanks{Email: \texttt{cosimo.munari@math.ethz.ch}}}
\affil{Department of Mathematics, ETH Zurich, Switzerland}

\date{January 29, 2014}

\maketitle

\begin{abstract}
\noindent The theory of acceptance sets and their associated risk measures plays a key role in the design of capital adequacy tests. The objective of this paper is to investigate, in the context of bounded financial positions, the class of \textit{surplus-invariant} acceptance sets. These are characterized by the fact that acceptability does not depend on the positive part, or \textit{surplus}, of a capital position. We argue that surplus invariance is a reasonable requirement from a regulatory perspective, because it focuses on the interests of liability holders of a financial institution. We provide a dual characterization of surplus-invariant, convex acceptance sets, and show that the combination of surplus invariance and coherence leads to a narrow range of capital adequacy tests, essentially limited to scenario-based tests. Finally, we emphasize the advantages of dealing with surplus-invariant acceptance sets as the primary object rather than directly with risk measures, such as {\em loss-based} and {\em excess-invariant risk measures}, which have been recently studied by Cont, Deguest \& He in~\cite{ContDeguestHe2013} and by Staum in~\cite{Staum2013}, respectively.
\end{abstract}


\MSCclassification{91B30, 91B32}

\JELclassification{C60, G11, G22}

\parindent 0em \noindent

%%%%%%%%%%%%%%%%%%%%%%%%%%%%%%%%%%%%%%%%%%%%%%%%%%%%%%%%%%%%%%%%%%%%%%%%%%%%%%%%%%%%%%%%%%%%%%%%%%%%%%%%%%%%%%%%%%%%%%%%%%

\section{Introduction}

The theory of acceptance sets and risk measures occupies an important place in current debates about solvency regimes in both the insurance and the banking world. The objective of this paper is to investigate the class of \textit{surplus-invariant} acceptance sets and their associated risk measures. These acceptance sets are characterized by the fact that acceptability does not depend on the positive part, or surplus, of a capital position. Further down we argue that surplus-invariant acceptance sets and risk measures have properties that are natural from the perspective of {\em external risk measures}, a term coined in Kou, Peng \& Heyde~\cite{KouPengHeyde2013} to denote risk measures used for external regulation rather than for purely internal risk management purposes. In fact, risk measures with related properties have been recently introduced and studied independently by Cont, Deguest \& He in~\cite{ContDeguestHe2013}, who called them \textit{loss based} and by Staum in~\cite{Staum2013}, who use the term \textit{excess invariant}.

\medskip

Coherent risk measures were introduced by Artzner, Delbaen, Eber \& Heath in the seminal paper~\cite{ArtznerDelbaenEberHeath1999} for finite sample spaces and by Delbaen in~\cite{Delbaen2002} for general probability spaces. Convex risk measures were then studied by F\"{o}llmer \& Schied in~\cite{FoellmerSchied2002} and by Frittelli \& Rosazza Gianin in~\cite{FrittelliGianin2002}. More recently,  Artzner, Delbaen \& Koch-Medina~\cite{ArtznerDelbaenKoch2009} and Farkas, Koch-Medina \& Munari~\cite{FarkasKochMunari2013a} and \cite{FarkasKochMunari2013b} have emphasized the importance of viewing acceptance sets as the primitive objects and risk measures as a particular way of measuring the ``distance'' to acceptability. In keeping with this view we adopt the framework developed in~\cite{FarkasKochMunari2013a} and~\cite{FarkasKochMunari2013b}.

\subsubsection*{Surplus-invariant acceptance sets}

Throughout this paper \textit{capital positions} -- assets net of liabilities -- of financial institutions are represented by the elements of $L^\infty$, the space of essentially bounded random variables on a probability space $(\Omega,\cF,\probp)$. An institution is said to be adequately capitalized if its capital position belongs to a pre-specified \textit{acceptance set}, i.e. to a nonempty, proper subset $\cA$ of $L^\infty$ such that $Y\in\cA$ whenever $Y\geq X$ almost surely for some $X\in\cA$.

\medskip

For a financial institution with limited liability, the positive and negative parts of a capital position $X$ have a clear financial interpretation. The positive part $X^+:=\max\{X,0\}$ is called the {\em surplus} and represents the excess of available funds over the amount needed to meet liabilities. The negative part $X^-:=\max\{-X,0\}$ represents the amount over which available funds fall short of meeting liabilities. Since $X^-$ reflects the limited liability of the owners of the institution, it is often referred to as the {\em owners' option to default}. As explained in more detail in Section~2, when assessing the capital adequacy of a financial institution, regulators -- who act on behalf of liability holders -- should mainly focus on the option to default. This is the rationale for investigating \textit{surplus-invariant} acceptance sets, i.e. acceptance sets $\cA\subset L^\infty$ that have the following property:
\begin{equation}
\label{intro: surplus inv property acceptance sets}
X\in\cA \textrm{ and } Y^-=X^- \textrm{ imply } Y\in\cA\,.
\end{equation}
In words, if a position is acceptable, then any other position having the same negative part is also acceptable, \textit{regardless of its surplus}, i.e. acceptability is not driven by an excess of funds but depends exclusively on the profile of the downside of the capital position.

\subsubsection*{Risk measures and surplus invariance}

While in this paper we take acceptance sets as the point of departure, the authors of~\cite{ContDeguestHe2013} and~\cite{Staum2013} address the issue of surplus invariance within the framework of {\em positive risk measures}, i.e. decreasing functions $\rho:L^\infty \to [0,\infty)$ satisfying the normalization condition $\rho(0)=0$. When viewed as a capital requirement, the number $\rho(X)$ is interpreted as the minimum amount of additional capital that needs to be raised to make the position acceptable. The focus of~\cite{ContDeguestHe2013} and~\cite{Staum2013} is on positive risk measures that are surplus invariant in the following sense: If the capital position $X\in L^\infty$ of a financial institution is not acceptable, then the capital required to make it acceptable does not depend on the surplus $X^+$, i.e.
\begin{equation}
\label{intro: surplus inv property}
\rho(X)=\rho(-X^-) \ \ \mbox{for all} \ X\in L^\infty \,.
\end{equation}

\medskip

Even though, at this level of generality, a positive risk measure $\rho$ does not explicitly build on an acceptance set, the following notion of ``acceptability'' is implicit in its interpretation as a capital requirement: A position $X\in L^\infty$ is {\em acceptable} if it does not require additional capital, i.e. if it belongs to the set
\begin{equation}
\cA(\rho):=\{X\in L^\infty \,; \ \rho(X)=0\}\,,
 \end{equation}
which is easily seen to be an acceptance set. However, what is neither explicit nor implicit is the operational meaning of $\rho$. This is a critical aspect since in any operational setting it is not only important to know how much capital to raise but also how the raised capital can be used to improve the firm's capital position. For this reason we believe that, in the context of capital adequacy, it is important to always start from a notion of ``acceptability'' and only then consider the associated risk measures, as advocated in~\cite{ArtznerDelbaenKoch2009},  in~\cite{FarkasKochMunari2013a} and~\cite{FarkasKochMunari2013b}. Here, we focus on risk measures $\rho_{\cA,S}:L^\infty\to\overline{\R}$ defined as
\begin{equation}
\label{definition: risk measure intro 1}
\rho_{\cA,S}(X):=\inf\left\{m\in\R \,; \ X+\frac{m}{S_0}S_T\in\cA\right\} \ \ \ \mbox{for} \ X\in L^\infty\,,
\end{equation}
where $\cA\subset L^\infty$ is a given acceptance set and $S=(S_0,S_T)$ is a traded asset with initial price $S_0>0$ and positive payoff $S_T\in L^\infty$. The quantity $\rho_{\cA,S}(X)$ has an explicit operational meaning. Indeed, when finite and positive, it represents the amount of capital that needs to be raised and invested in the asset $S$ to make the capital position $X$ acceptable. When finite an negative, it represents the maximum amount of capital that can be extracted from the institution without compromising its capital adequacy.

When $\cA$ is surplus invariant, the risk measure $\rho_{\cA,S}$ enjoys the following {\em surplus invariance} property:
\begin{equation}
\rho_{\cA,S}(X)=\rho_{\cA,S}(-X^-) \ \ \mbox{for all} \ X\in L^\infty \ \mbox{such that} \ \rho_{\cA,S}(X)\geq0\,.
\end{equation}
Note that the above equality requires that $\rho_{\cA,S}(X)\geq0$. Thus, for unacceptable positions with the same negative part capital requirements are identical. By contrast, for acceptable positions, the amount of capital that can be extracted without compromising acceptability will typically not depend only on the negative part, but also on how far into acceptability the position is.

\medskip

For positive risk measures, condition $\rho_{\cA,S}(X)\geq0$ is automatically satisfied for every $X\in L^\infty$. This is because said risk measures do not provide any information regarding the amount of capital that could be extracted while retaining acceptability. The main argument for ignoring this dimension seems to be the conviction that a regulator should only be interested in the case where a  positive capital injection is needed, i.e. in improving the capital adequacy of a badly capitalized institution. However, neglecting the question of whether there is room to extract capital from an adequately capitalized institution is problematic from the perspective of both managers and supervisors of a financial institution. Indeed, financial institutions have a strong interest in managing their capital base efficiently, i.e. they have an interest in knowing whether they can give capital back and, if so, how much. On the other hand, for the supervisors of an institution, it is important to know whether the institution is navigating at the limit of acceptability, in which case insolvency may be imminent, or whether it has a reasonable buffer that would justify less intensive supervision. For this reason, we believe that from a capital adequacy perspective, it is more natural to study risk measures that can also attain negative values.

\subsubsection*{Main contribution}

After presenting several examples of surplus-invariant acceptance sets in Section \ref{section: Surplus-invariant acceptance sets and risk measures}, we focus on surplus-invariant acceptance sets that are {\em convex} or {\em coherent}, i.e. convex cones. Convex and coherent acceptance sets are important because they reflect the principle that higher diversification should lead to lower capital requirements. In Theorem \ref{dual repr convex xs acc} we provide a dual characterization of convex, surplus-invariant acceptance sets that are closed with respect to the weak$^\ast$ topology on $L^\infty$. This topological property is commonly assumed across the literature because the corresponding risk measures $\rho_{\cA,S}$ then satisfy the Fatou property so that their dual representations involve probability measures rather than finitely-additive measures. Moreover, if $(\Omega,\cF,\probp)$ is nonatomic, it is well-known, e.g.~\cite{Svindland2010}, that every law-invariant, convex acceptance set that is norm closed is also weak$^\ast$ closed. As a result, we also obtain corresponding dual representations for convex, surplus-invariant risk measures satisfying the Fatou property.

\smallskip

Using the aforementioned dual representation we show in Theorem \ref{theorem coherent} that the only weak$^\ast$-closed, coherent acceptance sets that are surplus invariant are {\em scenario-based} acceptance sets, i.e. acceptance sets of the form\footnote{The acronym SPAN stands for Standard Portfolio ANalysis and refers to a methodology used to compute margin requirements adopted by many central exchanges.}
\begin{equation}
\SPAN(A):=\{X\in L^\infty \,; \ X1_A\ge 0 \ \text{almost surely}\}
\end{equation}
for some scenario set $A\in\cF$. Therefore, if we simultaneously insist on coherence and surplus invariance, we end up with an extremely limited choice of acceptability criteria that have properties which may not always be desirable. Indeed, if $A\neq\Omega$, then $\SPAN(A)$ is blind to the behaviour of a capital position outside $A$ and, if $A=\Omega$, we end up requiring a zero default probability for an institution to be acceptable. Hence, if we view surplus invariance as a desirable property, to obtain a wider range of acceptability criteria we need to abandon coherence. The situation is even more extreme if law invariance is additionally required: The only weak$^\ast$-closed, coherent acceptance set that is simultaneously law invariant and surplus invariant is the positive cone $L^\infty_+$. In particular, acceptance sets based on Tail Value-at-Risk are not surplus invariant. It is worthwhile noting that acceptance sets based on Value-at-Risk, the other standard acceptability test used in practice are surplus invariant, though not convex. This discussion has implications for policy making since it makes manifest that the choice of a regulatory risk measure is not straightforward and tradeoffs need to be made.

\subsubsection*{Loss-based and excess-invariant risk measures}

The last section of the paper is devoted to clarifying the link between surplus-invariant acceptance sets and the type of risk measures studied in Cont, Deguest \& He~\cite{ContDeguestHe2013} and in Staum~\cite{Staum2013}. These risk measures are positive and surplus invariant. Our analysis targets those properties we deem most relevant from a capital adequacy perspective: The associated acceptability criteria and the corresponding operational interpretation.

\smallskip

We show in Section \ref{subsection: accept and oper meaning} that the acceptability criterion linked to all the examples provided in~\cite{ContDeguestHe2013}, with the exception of scenario-based risk measures, is the most restrictive one: A position $X\in L^\infty$ requires no capital if and only if $X\ge 0$ almost surely. This seems to be too restrictive from a capital adequacy perspective, because the only adequately capitalized institutions would be those having a default probability equal to zero.

\smallskip

Contrary to risk measures of the form $\rho_{\cA,S}$, a positive, surplus-invariant risk measure $\rho$ does not possess an obvious operational interpretation. A natural way to vest $\rho$ with an operational meaning, which is also mentioned in~\cite{ContDeguestHe2013}, is to require $\rho$ to be {\em cash compatible}:
\begin{equation}
\rho(X+\rho(X))=0 \ \ \ \mbox{for all} \ X\in L^\infty\,.
\end{equation}
Under this assumption, adding the capital amount $\rho(X)$ and holding it in {\em cash} would be sufficient to ensure acceptability. Not all positive, surplus-invariant risk measures satisfy this operational property; for instance, the put-option premium studied by Jarrow in~\cite{Jarrow2002} and considered in~\cite{ContDeguestHe2013} is not cash compatible. In Theorem \ref{repr rho vee under self balancing} we show that every positive, surplus-invariant risk measure $\rho$ that is cash compatible and satisfies the axiom of cash subadditivity introduced in El Karoui \& Ravanelli~\cite{ElKarouiRavanelli2009} can be written as
\begin{equation}
\label{intro: truncation}
\rho(X)=\max\{\rho_{\cA,S}(X),0\} \ \ \ \mbox{for} \ X\in L^\infty,
\end{equation}
where $\cA\subset L^\infty$ is surplus invariant and $S=(1,1_\Omega)$ is the cash asset. Since the loss-based risk measures in~\cite{ContDeguestHe2013} are always cash subadditive when convex, they can be represented as truncated risk measures of the form $\rho_{\cA,S}$ as in \eqref{intro: truncation} whenever they are cash compatible. The same holds true for every risk measure that is cash additive subject to positivity, as considered in \cite{Staum2013}. This highlights, once more, the importance of studying risk measures of the form $\rho_{\cA,S}$ associated with surplus-invariant acceptance sets.

%%%%%%%%%%%%%%%%%%%%%%%%%%%%%%%%%%%%%%%%%%%%%%%%%%%%%%%%%%%%%

\section{Capital adequacy tests and risk measures}
\label{section framework}

We consider a one-period economy with dates $t=0$ and $t=T$. Uncertainty at time $t=T$ is modeled by a probability space $(\Omega,\cF,\probp)$. By $L^\infty$ we will denote the Banach space of all real-valued random variables $X$ on $(\Omega,\cF)$ that are almost surely bounded. As usual, two random variables are identified if they coincide almost surely. If $A\in\cF$ we denote its indicator function by $1_A$. The space $L^\infty$ becomes a Banach lattice when equipped with the natural ordering, i.e. for $X,Y\in L^\infty$ we have $Y\ge X$ whenever this inequality holds almost surely. A random variable $X\in L^\infty$ is said to be positive if $X\ge 0$. The positive cone $L^\infty_+$ is the closed, convex cone consisting of all positive random variables. Given $X\in L^\infty$, we set $X^+:=\max\{X,0\}$ and $X^-:=\max\{-X,0\}$.

\medskip

In our economy, financial institutions issue liabilities and invest in assets at time $t=0$. They settle their liabilities using the payoff of the assets they hold at time $t=T$. Assets and liabilities are assumed to be denominated with respect to a fixed num\'{e}raire and to belong to $L^\infty$. Let $A\in L^\infty$ and $L\in L^\infty$ be the random variables representing the terminal payoff of the institution's assets and liabilities, respectively. The {\em net financial position} or {\em capital position} of the institution is the random variable
\begin{equation}
X:=A-L\in L^\infty\,.
\end{equation}

\medskip

Provided the owners of the institution have {\em unlimited liability}, $X$ represents what is left for them after liabilities have been settled. More typical is the situation in which the owners of the institution have {\em limited liability}. In this case, if the assets of the institution do not suffice to repay its liabilities, the institution defaults and the owners are under no obligation to finance the deficit. In such a case, the {\em owners' surplus} and the {\em owners' option to default} are, respectively,
\begin{equation}
K:=X^+ \ \ \ \mbox{and} \ \ \ D:=-X^-\,.
\end{equation}
In particular, the option to default $D$ represents the amount by which the institution will default at time $t=T$.

\medskip

There is an inherent conflict between the interests of the owners and those of the liability holders. The option to default gives the owners the right to walk away should the proceeds of the assets not suffice to meet the liabilities. Consequently, the owners will care about the possibility of default but not about its \textit{magnitude}, i.e. they will focus more on their surplus $K$ and less on the option to default $D$. On the other hand, liability holders will focus on the option to default, rather than the surplus, because they suffer the consequences of the former and derive no benefits from the latter. In fact, one of the key objectives of external regulation is to help control the probability and magnitude of possible defaults of financial institutions, thereby reducing to acceptable levels the negative implications for liability holders in case of default.

\medskip

Acceptance sets and risk measures are used to formalize the process of setting capital requirements. As usual, we denote by $\overline{\R}:=\R\cup\{\pm\infty\}$ the extended real line.

\medskip

\begin{definition}
\label{definition acceptance set}
A nonempty, proper subset $\cA\subset L^\infty$ is called an {\em acceptance set} if $X\in\cA$ and $Y\ge X$ imply $Y\in\cA$. An acceptance set that is a convex cone is said to be {\em coherent}.
\end{definition}

\medskip

Any acceptance set $\cA\subset L^\infty$ can be regarded as a \textit{capital adequacy test} specified by the regulator: An institution with capital position $X$ passes the capital adequacy test whenever $X\in\cA$.

\medskip

\begin{definition}
A map $\rho:L^\infty\to\overline{\R}$ is said to be a \textit{risk measure} if it is decreasing, i.e. if
\begin{equation}
\rho(Y)\leq\rho(X) \ \ \ \mbox{for any} \ Y\geq X\,,
\end{equation}
and if the set
\begin{equation}
\cA(\rho):=\{X\in L^\infty \,; \ \rho(X)\le 0\}
\end{equation}
is nonempty and proper.
\end{definition}

\medskip

\begin{remark}
For every risk measure $\rho:L^\infty\to\overline{\R}$, the set $\cA(\rho)$ is an acceptance set.
\end{remark}

\subsubsection*{Measuring the distance to acceptability}

To assign an operational meaning to a risk measure $\rho:L^\infty\to\overline{\R}$, we need to describe how to interpret the quantity $\rho(X)$. In this paper, the initial focus will be on risk measures $\rho_{\cA,S}$ associated with a pre-specified acceptance set $\cA\subset L^\infty$ and a pre-specified traded asset $S=(S_0,S_T)$ with initial value $S_0>0$ and nonzero terminal payoff $S_T\in L^\infty_+$. The asset $S$ is called the \textit{eligible} asset. These risk measures are designed to capture the following operational situation: To modify the acceptability of its capital position, the management of a financial institution will be allowed to raise capital and invest it in $S$. When positive, the number $\rho_{\cA,S}(X)$ represents the minimum amount of capital that needs to be raised an invested in $S$ to make the position $X$ acceptable. In this sense, $\rho_{\cA,S}(X)$ can be interpreted as a measure of the ``distance'' of $X$ to acceptability using the asset $S$ as the underlying ``yardstick''. The formal definition is as follows:

\medskip

\begin{definition}
\label{definition risk measure}
Let $\cA\subset L^\infty$ be an acceptance set and $S=(S_0,S_T)$ a traded asset. The risk measure associated to $\cA$ a and $S$ is the map $\rho_{\cA,S}:L^\infty\to\overline{\R}$ defined by
\begin{equation}
\label{risk measure intro}
\rho_{\cA,S}(X):=\inf\left\{m\in\R \,; \ X+\frac{m}{S_0}S_T\in\cA\right\} \ \ \ \mbox{for} \ X\in L^\infty\,.
\end{equation}
If $S=(1,1_\Omega)$ is the risk-free asset, we simply write
\begin{equation}
\rho_{\cA}(X):=\inf\left\{m\in\R \,; \ X+m\in\cA\right\} \ \ \ \mbox{for} \ X\in L^\infty\,.
\end{equation}
\end{definition}

\medskip

It is easy to prove that every map $\rho_{\cA,S}$ is indeed a risk measure and satisfies the following $S$-additivity property.

\medskip

\begin{definition}
Let $S=(S_0,S_T)$ be a traded asset. A risk measure $\rho:L^\infty\to\overline{\R}$ is called \textit{$S$-additive} if
\begin{equation}
\label{S additivity}
\rho(X+\lambda S_T)=\rho(X)-\lambda S_0 \ \ \ \mbox{for all} \ X\in  L^\infty \ \mbox{and} \ \lambda\in\R\,.
\end{equation}
If $\rho$ is $S$-additive with respect to $S=(1,1_\Omega)$, then $\rho$ is called \textit{cash additive}.
\end{definition}

\medskip

\begin{remark}
\label{general remark S additivity}
\begin{enumerate}[(i)]
\item Let $S=(S_0,S_T)$ be a traded asset. It is straightforward to prove that if a risk measure $\rho:L^\infty\to\overline{\R}$ is $S$-additive, then $\rho=\rho_{\cA(\rho),S}$.
\item If~$S=(S_0,S_T)$ is a general traded asset, then the corresponding $S$-additive risk measures need not be finitely valued nor continuous. Conditions for finiteness and continuity can be found in~\cite{FarkasKochMunari2013a} and~\cite{FarkasKochMunari2013b}. If $S_T$ is essentially bounded away from zero, however, then $\rho_{\cA,S}$ is always finitely valued and Lipschitz continuous. This is the case with cash-additive risk measures, for which we also refer to Chapter~4 in~\cite{FoellmerSchied2011}.
\end{enumerate}
\end{remark}

%%%%%%%%%%%%%%%%%%%%%%%%%%%%%%%%%%%%%%%%%%%%%%%%%%%

\section{Surplus-invariant acceptance sets and risk measures}
\label{section: Surplus-invariant acceptance sets and risk measures}

In this section we introduce surplus-invariant acceptance sets and risk measures arguing that, from a regulatory perspective, it is reasonable to use capital adequacy tests for which acceptability cannot be achieved by increasing the surplus of the financial institution.

\subsubsection*{Surplus-invariant acceptance sets}

\begin{definition}
\label{definition excess invariant acceptance set}
An acceptance set $\cA\subset L^\infty$ is called {\em surplus invariant} if $Y\in\cA$ whenever $Y^-=X^-$ for some $X\in\cA$.
\end{definition}

\medskip

Surplus invariance is an important property from the perspective of liability holders. Indeed, assume $\cA$ were not surplus invariant. We could then find positions $X\in\cA$ and $Y\not\in \cA$ with $X^-=Y^-$. However, liability holders should be indifferent to the institution having $X$ or $Y$: If the company were to default, then it would default in both cases by the same amount. Hence, what would make $X$ acceptable versus $Y$ would only be an excess of assets over liabilities, which would accrue to the owners of the company and not to liability holders. Thus, a disadvantage for the liability holders (the default) would be compensated by an advantage for the owners (the excess funds).

\medskip

\begin{remark}
Surplus-invariant acceptance sets in $L^\infty$ have been considered in~\cite{Staum2013} where they were called \textit{excess-invariant}. We have opted for the term ``surplus'' because it is more common in the insurance literature.
\end{remark}

\medskip

We start by providing an elementary, yet useful, characterization of surplus-invariant acceptance sets.

\medskip

\begin{lemma}
\label{lemma: xs invariance and indicators}
Let $\cA\subset L^\infty$ be an acceptance set. Then the following statements are equivalent:
\begin{enumerate}[(a)]
\item $\cA$ is surplus invariant;
\item if $X\in\cA$, then $X1_A\in\cA$ for every $A\in\cF$;
\item if $X\in\cA$, then $-X^-\in\cA$.
\end{enumerate}
In particular, any surplus-invariant acceptance set contains $0$.
\end{lemma}
\begin{proof}
Assume \textit{(a)} holds and let $X\in\cA$ and $A\in\cF$. Note that $-X^-\in\cA$ by surplus invariance. Then, since $-X^-\le X1_A,$ we obtain $X1_A\in\cA$ by monotonicity of $\cA$ so that \textit{(b)} holds. If \textit{(b)} holds and $X\in\cA$, it is sufficient to set $A:=\{X<0\}$ to have $-X^-= X1_A\in\cA$, proving \textit{(c)}. Finally, assume \textit{(c)} holds, and take $X\in\cA$ and $Y\in L^\infty$ such that $Y^-=X^-$. Then it follows immediately that $-Y^-\in\cA$, implying $Y\in\cA$ by monotonicity. Hence, $\cA$ is surplus invariant.
\end{proof}

\subsubsection*{Examples of surplus-invariant acceptance sets}

\begin{example}[$\SPAN$]
\label{span example}
Let $A$ be a measurable subset of $\Omega$. The $\SPAN$ acceptance set generated by $A$ is the set
\begin{equation}
\SPAN(A):=\{X\in L^\infty \,; \ X1_A\geq0\}\,.
\end{equation}
This is a closed, coherent acceptance set that is surplus invariant. As we show in Theorem~\ref{theorem coherent}, SPAN acceptance sets are essentially the only coherent, surplus-invariant acceptance sets. The acronym SPAN stands for Standard Portfolio ANalysis. It refers to a methodology used to compute margin requirements adopted by many central exchanges, see for instance the discussion in~\cite{ContDeguestHe2013}.
\end{example}

\medskip

\begin{example}[$\SPAN$-type]
\label{span-type example}
Let $A$ be a measurable subset of $\Omega$ and take $V\in L^\infty$. The $\SPAN$-{\em type} acceptance set generated by $A$ and $V$ is the set
\begin{equation}
\SPAN(A,V):=\{X\in L^\infty \,; \ X1_A\geq V1_A\}\,.
\end{equation}
This is a closed, convex acceptance set, which is surplus invariant if and only if $V1_A\le 0$.
\end{example}

\medskip

\begin{example}[Shortfall risk]
\label{shortfall-risk example}
Let $\ell:\R\to\R$ be a nonconstant, convex, decreasing function and fix a level $\alpha>\inf_{x\in\R}\ell(x)$. The set
\begin{equation}
\cA_\ell:=\{X\in L^\infty \,; \ \E[\ell(-X^-)]\leq\alpha\}
\end{equation}
is a closed, convex, surplus-invariant acceptance set, which is also law invariant.
\end{example}

\medskip

The two most widely used risk measures in practice are Value-at-Risk and Tail Value-at-Risk. A detailed treatment can be found in Section~4.4 in~\cite{FoellmerSchied2011}. While the acceptance set corresponding to Value-at-Risk is surplus invariant, this is generally not the case for acceptability criteria based on Tail Value-at-Risk.

\medskip

\begin{example}[$\VaR$]
\label{var example}
The {\em Value-at-Risk} of $X\in L^\infty$ at the level $\alpha\in(0,1)$ is defined as
\begin{equation}
\VaR_\alpha(X):=\inf\{m\in\R \,; \ \probp(X+m<0)\le\alpha\}\,.
\end{equation}
The associated acceptance set
\begin{equation}
\label{var acceptance set}
\cA_\alpha:=\{X\in L^\infty \,; \ \VaR_\alpha(X)\le 0\}=\{X\in L^\infty \,; \ \probp(X<0)\le\alpha\}
\end{equation}
is easily seen to be surplus invariant. It is a closed cone, which is, in general, not convex.
\end{example}

\medskip

Recall that an acceptance set $\cA\subset L^\infty$ is said to be {\em sensitive} if $X\not\in\cA$ for any nonzero $X\in L^\infty$ such that $X\leq0$.

\medskip

\begin{remark}
\label{xs-invariance and sensitivity remark}
The only sensitive, surplus-invariant acceptance set is $L^\infty_+$; see also Proposition~4.2 in~\cite{Staum2013}.
\end{remark}

\medskip

\begin{example}[$\TVaR$]
\label{tvarexample}
The {\em Tail Value-at-Risk} of $X\in L^\infty$ at the level $\alpha\in(0,1)$ is defined as
\begin{equation}
\TVaR_\alpha(X):=\frac{1}{\alpha} \int_0^\alpha\VaR_\beta(X)d\beta\,.
\end{equation}
Tail Value-at-Risk is also known under the names of {\em Expected Shortfall},
{\em Conditional Value-at-Risk}, or {\em Average Value-at-Risk}. The corresponding acceptance set
\begin{equation}
\cA^\alpha:=\{X\in L^\infty \,; \ \TVaR_\alpha(X)\le 0\}
\end{equation}
is a closed, coherent acceptance set. However, since $\cA^\alpha$ is well-known to be sensitive, it is not surplus invariant by Remark~\ref{xs-invariance and sensitivity remark} (except for the trivial cases where it coincides with $L^\infty_+$). This can also be seen directly by considering the following example: Take $A\in\cF$ with $0<\probp(A)<1$ and set
\begin{equation}
X:=-\e 1_A+\gamma 1_{A^c} \ \ \ \mbox{for} \ \e,\gamma>0\,.
\end{equation}
It is easy to show that $\VaR_\beta(X)=\e$ if $\beta<\probp(A)$ and $\VaR_\beta(X)=-\gamma$ otherwise. Take a TVaR level $\alpha$ such that $\probp(A)<\alpha<1$. Then
\begin{eqnarray}
\TVaR_\alpha(X) &=& \frac{1}{\alpha}\int^{\probp(A)}_{0}\VaR_\beta(X)d\beta + \frac{1}{\alpha}\int^{\alpha}_{\probp(A)}\VaR_\beta(X)d\beta \\
 &=& \frac{\e}{\alpha}\,\probp(A) + \gamma\left(\frac{\probp(A)}{\alpha}-1\right)\,.
\end{eqnarray}
From the same calculations we also obtain
\begin{equation}
\TVaR_\alpha(-X^-) = \frac{\e}{\alpha}\,\probp(A)>0\,.
\end{equation}
As a result, the negative part $-X^-$ is always unacceptable with respect to $\cA^\alpha$, whereas $X$ is acceptable whenever the surplus $\gamma$ is sufficiently large.
\end{example}

\subsubsection*{Surplus-invariant risk measures}

Not surprisingly, $S$-additive risk measures associated with surplus-invariant acceptance sets are essentially characterized by the following property: They assign the same amount of required capital to {\em unacceptable} positions having the same negative part. Based on this property, which is shown below, we define the notion of a surplus-invariant risk measure. The class of surplus-invariant risk measures comprises the loss-based risk measures introduced by Cont, Deguest \& He in~\cite{ContDeguestHe2013}, and the excess-invariant risk measures introduced by Staum in~\cite{Staum2013}.

\medskip

\begin{definition}
A risk measure $\rho:L^\infty\to\overline{\R}$ is called \textit{surplus invariant} if
\begin{equation}
\label{surplus inv rm}
\rho(X)=\rho(-X^-) \ \ \ \mbox{for all} \ X\in L^\infty \ \mbox{such that} \ \rho(X)\ge 0\,.
\end{equation}
\end{definition}

\medskip

\begin{remark}
\label{remark on xs invariant risk measures}
In line with the terminology introduced in~\cite{Staum2013}, a risk measure satisfying~\eqref{surplus inv rm} could also be called ``surplus invariant subject to positivity''. To keep language as light as possible we have refrained from doing this because the requirement $\rho(X)\geq0$ in~\eqref{surplus inv rm} is generally necessary. Indeed, assume $\rho$ is {\em normalized}, i.e. $\rho(0)=0$, and $\rho(X)=\rho(-X^-)$ for some $X\in L^\infty$. Then, by monotonicity, $0=\rho(0)\leq\rho(-X^-)=\rho(X)$.
\end{remark}

\medskip

\begin{proposition}
\label{xs-invariant risk measure}
Let $\cA\subset L^\infty$ be an acceptance set and $S=(S_0,S_T)$ a traded asset.
\begin{enumerate}[(i)]
\item If $\cA$ is surplus invariant, then $\rho_{\cA,S}$ is surplus invariant.
\item If $\cA$ is closed and $\rho_{\cA,S}$ is surplus invariant and does not attain the value $-\infty$, then $\cA$ is surplus invariant.
\end{enumerate}
\end{proposition}
\begin{proof}
\textit{(i)} Take $X\in L^\infty$. Since $-X^-\le X$ we clearly have $\rho_{\cA,S}(X)\leq\rho_{\cA,S}(-X^-)$. Assume now that $X+\lambda S_T\in\cA$ for some $\lambda>0$. By surplus invariance we then have $-(X+\lambda S_T)^-\in\cA$. Since $\lambda$ is positive we see that $-(X+\lambda S_T)^-\leq -X^-+\lambda S_T$. This implies $-X^-+\lambda S_T\in\cA$ and, consequently, $\rho_{\cA,S}(-X^-)\leq\lambda$. As a result, we obtain $\rho_{\cA,S}(-X^-)\leq\rho_{\cA,S}(X)$.

\smallskip

\textit{(ii)} Since $\cA$ is closed, we have $\cA=\{X\in L^\infty \,; \ \rho_{\cA,S}(X)\le 0\}$. Take $X\in\cA$ so that $\rho_{\cA,S}(X)\le 0$, and set $\lambda:=\rho_{\cA,S}(X)/S_0$. Since, $\rho_{\cA,S}(X+\lambda S_T )=0$ we have $\rho_{\cA,S}(-(X+\lambda S_T)^-)=0$ by surplus invariance. This means that $-(X+\lambda S_T)^-\in\cA$. However, $-(X+\lambda S_T)^-\le -X^-$ since $\lambda\le 0$, hence $-X^-\in\cA$. It follows from Lemma~\ref{lemma: xs invariance and indicators} that $\cA$ is surplus invariant.
\end{proof}

\medskip

\begin{remark}
In the preceding proposition, requiring that $\rho_{\cA,S}$ does not attain the value $-\infty$ is reasonable. Indeed, any position $X\in L^\infty$ with $\rho_{\cA,S}(X)=-\infty$ would have the following pathological property: One would be able to withdraw arbitrary amounts of capital without compromising acceptability. Moreover, if $\cA$ is closed and convex, then $\rho_{\cA,S}$ is a convex and lower semicontinuous, which implies that $\rho_{\cA,S}$ cannot attain any real value if it attains the value $-\infty$ by Proposition~2.4 in Chapter~I in~\cite{EkelandTemam1999}.
\end{remark}

%%%%%%%%%%%%%%%%%%%%%%%%%%%%%%%%%%%%%%%%%%%%

\section{Dual representations}
\label{representation section}

In this section we provide a characterization of convex, surplus-invariant acceptance sets that are closed with respect to the weak$^\ast$ topology $\sigma(L^\infty,L^1)$. Here we denote by $L^1$ the Banach lattice of real-valued random variables $X$ on $(\Omega,\cF)$ such that $\E[|X|]<\infty$, where again random variables coinciding almost surely are identified. The positive cone in $L^1$ is denoted by $L^1_+$. Note that every element $Z\in L^1$ can be identified with a functional $\psi_Z:L^\infty\to\R$ via the pairing $\psi_Z(X):=\E[XZ]$.

\medskip

\begin{remark}
The interest in $\sigma(L^\infty,L^1)$-closed acceptance sets $\cA\subset L^\infty$ lies in the fact that the corresponding risk measures $\rho_{\cA,S}$ have the so-called \textit{Fatou property}, i.e. they are lower semicontinuous with respect to $\sigma(L^\infty,L^1)$. Furthermore, Svindland showed in~\cite{Svindland2010} that, if the underlying probability space is nonatomic, every acceptance set in $L^\infty$ which is convex, closed and law invariant is automatically $\sigma(L^\infty,L^1)$-closed (see also Jouini, Schachermeyer \& Touzi~\cite{JouiniSchachermayerTouzi2006}). Recall that a set $\cA\subset L^\infty$ is said to be \textit{law invariant} if, whenever $X$ and $Y$ have the same distribution, $X\in\cA$ implies $Y\in\cA$.
\end{remark}

\medskip

\begin{remark}
\label{surplus invariant weak-star lsc}
Proposition~\ref{xs-invariant risk measure} holds, with the same proof, for a $\sigma(L^\infty,L^1)$-closed acceptance set $\cA\subset L^\infty$ in the following form: if $\rho_{\cA,S}$ does not attain the value $-\infty$, then $\rho_{\cA,S}$ is surplus invariant if and only if $\cA$ is surplus invariant.
\end{remark}

\medskip

The {\em (lower) support function} of a subset $\cA\subset L^\infty$ is the map $\sigma_\cA: L^1\to\R\cup\{-\infty\}$ defined by
\begin{equation}
\sigma_\cA(Z):=\inf_{X\in\cA}\E[XZ]\,.
\end{equation}
The set $B(\cA):=\{Z\in L^1 \,; \ \sigma_\cA(Z)>-\infty\}$ is called the {\em barrier cone} of $\cA$.

\medskip

\begin{remark}
\label{barrier cone remark}
Let $\cA$ be a subset of $L^\infty$. Then $\sigma_\cA$ is superlinear and upper semicontinuous. If $\cA$ is a cone, then $B(\cA)=\{Z\in L^1 \,; \ \sigma_\cA(Z)=0\}$. Moreover, by Lemma~3.11 in~\cite{FarkasKochMunari2013b}, if $\cA$ is an acceptance set, then $B(\cA)\subset L^1_+$.
\end{remark}

\subsubsection*{Dual representation of surplus-invariant acceptance sets}

We first establish a dual characterization of convex, surplus-invariant acceptance sets that are $\sigma(L^\infty,L^1)$-closed. Clearly, a characterization of this type also provides an explicit way to construct such acceptance sets.

\medskip

\begin{theorem}
\label{dual repr convex xs acc}
A subset $\cA$ of $L^\infty$ is a $\sigma(L^\infty,L^1)$-closed, convex, surplus-invariant acceptance set if and only if
\begin{equation}
\label{formula dual repr convex xs acc}
\cA=\bigcap_{Z\in B(\cA)}\{X\in L^\infty \,; \ \E[XZ]\geq\gamma(Z)\}
%=\bigcap_{Z\in L^1_+}\{X\in L^\infty \,; \ \E[XZ]\geq\gamma(Z)\}
\end{equation}
for some function $\gamma:L^1\to\R\cup\{-\infty\}$ satisfying
\begin{enumerate}[(F1)]
\item $\gamma(Z)\leq0$ for all $Z\in L^1_+$;
\item $\gamma(Z)>-\infty$ for some $Z\in L^1_+$;
\item $\gamma$ is decreasing.
\end{enumerate}
In~\eqref{formula dual repr convex xs acc} we may always choose $\gamma=\sigma_\cA$. Furthermore, for any $\gamma$ satisfying~\eqref{formula dual repr convex xs acc}, we have $\gamma(Z)\leq\sigma_\cA(Z)$ for all $Z\in L^1_+$. Moreover,
\begin{equation}
\label{support function for xs acc}
\sigma_\cA(Z)=\inf_{X\in\cA}\E[-X^-Z] \ \ \ \mbox{for every} \ Z\in L^1_+\,.
\end{equation}
\end{theorem}
\begin{proof}
%Since the second equality in \eqref{formula dual repr convex xs acc} is clear by Remark~\ref{barrier cone remark}, we only focus on the first one.
First, let $\cA(\gamma)$ be the intersection in~\eqref{formula dual repr convex xs acc}. As an intersection of $\sigma(L^\infty,L^1)$-closed halfspaces generated by positive, $\sigma(L^\infty,L^1)$-continuous, linear functionals, $\cA(\gamma)$ is a $\sigma(L^\infty,L^1)$-closed, convex acceptance set. In particular, $\cA(\gamma)$ is a nonempty, proper subset of $L^\infty$ by \textit{(F1)} and \textit{(F2)}. To prove that $\cA(\gamma)$ is surplus invariant, take $X\in\cA(\gamma)$. Since $\gamma$ is decreasing, for every $Z\in L^1_+$ we obtain
\begin{equation}
\E[-X^- Z] = \E[X1_{\{X<0\}} Z] \geq \gamma(1_{\{X<0\}} Z) \geq \gamma(Z)
\end{equation}
showing that $\cA(\gamma)$ is surplus invariant by Lemma~\ref{lemma: xs invariance and indicators}.

\smallskip

Now assume $\cA$ is a $\sigma(L^\infty,L^1)$-closed, convex acceptance set that is surplus invariant. It follows from Theorem~4.4 in~\cite{FarkasKochMunari2013c} that we can represent $\cA$ as in~\eqref{formula dual repr convex xs acc} taking $\gamma=\sigma_\cA$ and that, for any function $\gamma$ such that~\eqref{formula dual repr convex xs acc} is satisfied, we must have $\gamma(Z)\leq\sigma_\cA(Z)$ for all $Z\in L^1_+$. Note that the representation of $\sigma_\cA$ in~\eqref{support function for xs acc} follows from the surplus invariance of $\cA$. Hence, $\sigma_\cA$ satisfies \textit{(F3)}. Since $\cA$ is nonempty and proper, $\sigma_\cA$ also satisfies \textit{(F1)} and \textit{(F2)}.
\end{proof}

\medskip

\begin{remark}
\label{remark: radon nikodym}
Let $\cA\subset L^\infty$ be a $\sigma(L^\infty,L^1)$-closed, convex, surplus-invariant acceptance set. By positive homogeneity, we can restrict the intersection in \eqref{formula dual repr convex xs acc} to all random variables $Z\in B(\cA)$ satisfying $\E[Z]=1$. Random variables of this type can be identified with Radon-Nikodym derivatives $\frac{d\probq}{d\probp}$ of probability measures $\probq\ll\probp$, i.e. probability measures that are absolutely continuous with respect to $\probp$. Hence, the ``probabilistic'' version of the dual representation of $\cA$ in~\eqref{formula dual repr convex xs acc} is
\begin{equation}
\label{formula dual repr convex xs acc probabilities}
\cA=\bigcap_{\probq\ll\probp, \,\frac{d\probq}{d\probp}\in B(\cA)}\left\{X\in L^\infty \,; \ \E_\probq[X]\geq\sigma_\cA\left(\frac{d\probq}{d\probp}\right)\right\}\,.
\end{equation}
\end{remark}

\subsubsection*{Dual representation of surplus-invariant risk measures}

Next, we establish the consequences of surplus invariance for the standard dual representation of convex risk measures satisfying the Fatou property. For any $\rho:L^\infty\to\R\cup\{\infty\}$ we denote by $\rho^\ast$ its {\em $\sigma(L^\infty,L^1)$-conjugate}, i.e.
\begin{equation}
\rho^\ast(Z) := \sup_{X\in L^\infty}\{E[XZ]-\rho(X)\} \ \ \mbox{for} \ Z\in L^1\,.
\end{equation}

\medskip

\begin{proposition}
Let $\rho:L^\infty\to\R\cup\{\infty\}$ be a convex risk measures satisfying the Fatou property. If $\rho$ is surplus invariant, then
\begin{equation}
\label{dual repr surplus inv general rm}
\rho(X)=\sup_{\probq\ll\probp}\left\{-\E_\probq[X]-\rho^\ast\left(-\frac{d\probq}{d\probp}\right)\right\} \ \ \ \mbox{for} \ X\in L^\infty\,.
\end{equation}
Moreover, if $\rho(X)\geq0$, then
\begin{equation}
\label{dual repr surplus inv general rm, 2}
\rho(X)=\sup_{\probq\ll\probp}\left\{\E_\probq[X^-]-\rho^\ast\left(-\frac{d\probq}{d\probp}\right)\right\}\,.
\end{equation}
\end{proposition}
\begin{proof}
Since $\rho$ is convex and lower semicontinuous with respect to $\sigma(L^\infty,L^1)$, by combining Theorem~6 and Corollary~7 in~\cite{FrittelliGianin2002}, we obtain for any $X\in L^\infty$ the standard dual representation
\begin{equation}
\rho(X)=\sup_{Z\in L^1_+}\left\{-\E[XZ]-\rho^\ast\left(-Z\right)\right\}\,.
\end{equation}
Hence,~\eqref{dual repr surplus inv general rm} follows immediately if we pass to the corresponding Radon-Nikodym derivatives as discussed in Remark~\ref{remark: radon nikodym}. If $\rho(X)\geq0$, then $\rho(X)=\rho(-X^-)$ by surplus invariance, and therefore~\eqref{dual repr surplus inv general rm, 2} also holds.
\end{proof}

\medskip

For surplus-invariant risk measures of the form $\rho_{\cA,S}$, which are the operationally relevant ones, we obtain a sharper dual representation for which the impact of surplus invariance is more visible.

\begin{proposition}
\label{corollary dual repr surplus inv rm}
Let $\cA\subset L^\infty$ be a convex, $\sigma(L^\infty,L^1)$-closed acceptance set, and assume $\rho_{\cA,S}$ does not attain the value $-\infty$. If $\cA$ is surplus invariant, then
\begin{equation}
\label{dual repr surplus inv rm}
\rho_{\cA,S}(X)=\sup_{\probq\ll\probp, \,\E_\probq[S_T]=S_0}\left\{-\E_\probq[X]+\inf_{Y\in\cA}\E_\probq[-Y^-]\right\} \ \ \ \mbox{for} \ X\in L^\infty\,.
\end{equation}
Moreover, if $\rho_{\cA,S}(X)\geq0$, then
\begin{equation}
\label{dual repr surplus inv rm, 2}
\rho_{\cA,S}(X)=\sup_{\probq\ll\probp, \,\E_\probq[S_T]=S_0}\left\{\E_\probq[X^-]+\inf_{Y\in\cA}\E_\probq[-Y^-]\right\}\,.
\end{equation}
\end{proposition}
\begin{proof}
Since $\rho_{\cA,S}$ does not attain the value $-\infty$, we can apply Corollary 4.14 and Theorem 4.16 in~\cite{FarkasKochMunari2013c} to obtain
\begin{equation}
\rho_{\cA,S}(X)=\sup_{Z\in B(\cA), \,\E[S_TZ]=S_0}\left\{\sigma_\cA(Z)-\E[XZ]\right\}
\end{equation}
for every $X\in L^\infty$. Hence, the representation of $\sigma_\cA$ in~\eqref{support function for xs acc} yields~\eqref{dual repr surplus inv rm} once we consider the corresponding Radon-Nikodym derivatives. The second representation follows immediately from surplus invariance.
\end{proof}

%%%%%%%%%%%%%%%%%%%%%%%%%%%%%%%%%%%%%%%%%%%%%%%%%%%%%%%%%%%%%%%%%%%%%%%%%%%%%%%%%%%%%%%%%%%%%%%%%%%%%%%%%%%%%%%%%%%%%%%%%%%%%%%%%%%%%%%%%%%
\subsubsection*{Coherent, surplus-invariant acceptance sets}
%%%%%%%%%%%%%%%%%%%%%%%%%%%%%%%%%%%%%%%%%%%%%%%%%%%%%%%%%%%%%%%%%%%%%%%%%%%%%%%%%%%%%%%%%%%%%%%%%%%%%%%%%%%%%%%%%%%%%%%%%%%%%%%%%%%%%%%%%%%

We now focus on coherent, surplus-invariant acceptance sets and show that this class essentially coincides with the class of SPAN acceptance sets defined in Example~\ref{span example}. Hence, if surplus invariance is deemed to be a desirable property, to have a wider range of acceptability criteria we are forced to leave the coherent world. This becomes even more manifest if we additionally require law invariance, a common assumption in the literature. Indeed, we show that the positive cone $L^\infty_+$ is the only surplus-invariant coherent acceptance set that is also law invariant.

\medskip

In this section we assume that the space $L^1$ is separable. By Theorem~19.2 in~\cite{Billingsley1995}, this is the case whenever the $\sigma$-algebra $\cF$ is countably generated.

\medskip

\begin{lemma}
\label{support for coherent xs-invariant}
If $\cA\subset L^\infty$ is a coherent, surplus-invariant acceptance set that is $\sigma(L^\infty,L^1)$ closed, then
\begin{equation}
\label{xs-coherent barrier}
B(\cA)=\{Z\in L^1_+ \,; \ \E[X^-Z]=0, \ \forall \ X\in\cA\}\,.
\end{equation}
\end{lemma}
\begin{proof}
It follows from the representation of $\sigma_\cA$ in~\eqref{support function for xs acc} that for all $X\in\cA$ and $Z\in L^1_+$
\begin{equation}
\sigma_\cA(Z) = \inf_{Y\in\cA}\E[-Y^- Z] \leq -\E[X^- Z] \leq 0\,.
\end{equation}
Since $B(\cA)=\{Z\in L^1_+ \,; \ \sigma_\cA(Z)=0\}$ by Remark~\ref{barrier cone remark}, we immediately obtain~\eqref{xs-coherent barrier}.
\end{proof}

\medskip

\begin{theorem}
\label{theorem coherent}
Let $\cA$ be a coherent, surplus-invariant acceptance set in $L^\infty$ that is closed with respect to the $\sigma(L^\infty,L^1)$ topology. Then there exists $A\in\cF$ such that
\begin{equation}
\label{equation coherent xs}
\cA=\SPAN(A)=\{X\in L^\infty \,; \ X1_A\geq0\}\,.
\end{equation}
\end{theorem}
\begin{proof}
The acceptance set $\cA$ can be represented as in~\eqref{formula dual repr convex xs acc} taking $\gamma:=\sigma_\cA$. As a consequence of Lemma~\ref{support for coherent xs-invariant}, it is easy to show that
\begin{equation}
\label{dual repr coherent xs acc}
\cA = \bigcap_{Z\in B(\cA)}\{X\in L^\infty \,; \ \E[X^-Z]=0\}\,.
\end{equation}
Since $L^1$ is separable, Corollary 3.5 in~\cite{AliprantisBorder2006} implies there exists a countable, dense subset $(Z_n)$ of $B(\cA)$. We claim that
\begin{equation}
\label{dual repr coherent xs acc via separab}
\cA = \bigcap_{n\in\N}\{X\in L^\infty \,; \ \E[X^-Z_n]=0\}\,.
\end{equation}
We only need to prove the inclusion ``$\supset$'' since the converse inclusion follows immediately from the representation~\eqref{dual repr coherent xs acc}. To this end, take $X\in L^\infty$ satisfying $\E[X^-Z_n]=0$ for all $n\in\N$ and let $Z\in B(\cA)$. By density, there exists a subsequence $(Y_n)$ of $(Z_n)$ converging to $Z$ in $L^1$. This implies that $X^-Y_n\to X^-Z$ in $L^1$, showing that $\E[X^-Z]=0$. As a result, it follows from~\eqref{dual repr coherent xs acc} that $X$ belongs to $\cA$, thus the dual representation of $\cA$ in~\eqref{dual repr coherent xs acc via separab} holds. Finally, setting $A:=\bigcup_{n\in\N}\{Z_n>0\}\in\cF$ we obtain
\begin{equation}
\cA = \bigcap_{n\in\N}\{X\in L^\infty \,; \ \probp(\{X<0\}\cap\{Z_n>0\})=0\} = \{X\in L^\infty \,; \ X1_A\geq0\},
\end{equation}
which concludes the proof.
\end{proof}

\medskip

Since the positive cone $L^\infty_+$ is the only SPAN acceptance set that is law invariant, we immediately obtain the following result.

\medskip

\begin{corollary}
Let $\cA$ be a coherent, surplus-invariant acceptance set in $L^\infty$ that is closed with respect to the $\sigma(L^\infty,L^1)$ topology. If $\cA$ is law invariant, then $\cA=L^\infty_+$.
\end{corollary}

\medskip

\begin{remark}
The convex acceptance set based on shortfall risk introduced in Example~\ref{shortfall-risk example} is both surplus invariant and law invariant. This shows that requiring a convex, surplus-invariant acceptance set to be also law invariant is not as restrictive as in the coherent case.
\end{remark}

%%%%%%%%%%%%%%%%%%%%%%%%%%%%%%%%%%%%%%%%%%%%%%%%%%%%%%%%%%%%%%%%%%%%%%%%%%%%%%%%%%%%%%%%%%%%%%%%%%%%%%%%%%%%%%%%%%%%%%%%%%%%%%%%%%%%%%%%%%%
\section{Surplus invariance in finite dimensions}
\label{section: finite dimension}
%%%%%%%%%%%%%%%%%%%%%%%%%%%%%%%%%%%%%%%%%%%%%%%%%%%%%%%%%%%%%%%%%%%%%%%%%%%%%%%%%%%%%%%%%%%%%%%%%%%%%%%%%%%%%%%%%%%%%%%%%%%%%%%%%%%%%%%%%%%

An important case, in particular from a practical point of view, is the case where $\Omega$ is finite and $\cF$ is the discrete $\sigma$-algebra, i.e. the power set of $\Omega$. In this case $L^\infty$ can be identified with $\R^N$, where $N$ is the cardinality of $\Omega$. We show that closed, convex, surplus-invariant acceptance sets in $\R^N$ are essentially \textit{translates} of the positive cone $\R^N_+$. We start by adapting Lemma \ref{lemma: xs invariance and indicators} to the case $\Omega$ finite, omitting the proof.

\medskip

\begin{lemma}
\label{lemma finite 1}
An acceptance set $\cA\subset\R^N$ is surplus invariant if and only if $(\delta_1 x_1,\dots,\delta_N x_N)\in\cA$ for every $x\in\cA$ and $(\delta_1,\dots,\delta_N)\in\{0,1\}^N$.
\end{lemma}

\medskip

For the remainder of this section we assume $\cA\subset\R^N$ is a closed, surplus-invariant acceptance set. Define for every $j\in\{1,\cdots,N\}$
\begin{equation}
v_j:=\inf_{x\in\cA}x_j\in [-\infty,0]\,.
\end{equation}
We denote by $e_j$, $j\in\{1,\cdots,N\}$, the canonical basis vectors of $\R^N$.

\medskip

\begin{lemma}
\label{unbounded components}
Assume $v_{j_1}=\dots=v_{j_R}=-\infty$ for $j_1,\dots,j_R\in\{1,\dots,N\}$ and  $R\in\{1,\dots N\}$, then $\sum_{i=1}^R x_{j_i}e_{j_i}\in\cA$ for every $x_{j_1},\dots,x_{j_R}\in\R$.
\end{lemma}
\begin{proof} Take $x_{j_1},\cdots,x_{j_R}\in\R$. By assumption for every $i\in\{1,\dots,R\}$ we find $y\in\cA$ such that $y_i=Rx_{j_i}$. By Lemma~\ref{lemma finite 1}, this implies $Rx_{j_i}e_{j_i}\in\cA$ and the convexity of $\cA$ then yields $\sum_{i=1}^R x_{j_i}e_{j_i}\in\cA$, as claimed.
\end{proof}

\medskip

Note that the above implies that if $v_j=-\infty$ for every  $j\in\{1,\dots N\}$ then $\cA=\R^N$, which is impossible since $\cA$ is proper. As a consequence, there exists $j\in\{1,\dots,N\}$ such that $v_j>-\infty$. After renumbering, if necessary, we may assume that there exists $1\le K\le N$ such that $v_1,\dots,v_K\in\R$ and $v_{K+1}= \dots =v_N=-\infty$. We denote by $\pi:\R^N\to\R^K$ the canonical projection defined by
\begin{equation}
\pi(x_1,\dots,x_N):=(x_1,\dots,x_K)\,.
\end{equation}
By $(a,z):=(a_1,\dots,a_K,z_1,\dots,z_{N-K})$, $a\in\R^K$ and $z\in\R^{N-K}$, we denote a generic element of $\R^K\times\R^{N-K}$.

\medskip

\begin{lemma}
The set $\pi(\cA)$ is a closed, convex, surplus-invariant acceptance set in $\R^K$ and
\begin{equation}
\label{finite reduction}
\cA=\pi(\cA)\times\R^{N-K}\,.
\end{equation}
\end{lemma}
\begin{proof}
Clearly, $\pi(\cA)$ is a nonempty, proper subset of $\R^K$. Take any $a\in\pi(\cA)$ so that there exists $z\in\R^{N-K}$ such that $(a,z)\in\cA$. If $b\in\R^K$ with $b\ge a$, then $(a,z)\le (b,z)$ and we infer from the monotonicity of $\cA$ that $(b,z)\in\cA$. This in turn implies $b\in\pi(\cA)$, hence $\pi(\cA)$ is an acceptance set. Moreover, it is convex by the linearity of $\pi$. To see that $\pi(\cA)$ is surplus invariant, take $a\in\pi(\cA)$. By Lemma~\ref{lemma finite 1} we obtain $(a,0)\in\cA$. Since $\cA$ is surplus invariant, it follows that $(-a^-,0)\in\cA$, hence $-a^-\in\pi(\cA)$ showing that $\pi(\cA)$ is surplus invariant by Lemma \ref{lemma: xs invariance and indicators}. Finally, to prove that $\pi(\cA)$ is closed take a sequence $(a^n)$ in $\pi(\cA)$ converging to some $a\in\R^K$. Note that $(a^n,0)\in\cA$ again by Lemma~\ref{lemma finite 1}. Hence, since $\cA$ is closed, we must have $(a,0)\in\cA$ so that $a\in\pi(\cA)$.

\smallskip

Since $\cA\subset\pi(\cA)\times\R^{N-K}$ is obviously true, to prove~\eqref{finite reduction} we just need to establish the converse inclusion. To that effect take $(a,z)\in\pi(\cA)\times\R^{N-K}$. From Lemma~\ref{unbounded components} it follows that $(0,\alpha z)\in\cA$ for every $\alpha>0$. Moreover $(a,0)\in\cA$ by Lemma~\ref{lemma finite 1}. The convexity of $\cA$ now implies $((1-\lambda)a,\lambda\alpha z)\in\cA$ for every $\lambda\in(0,1)$ and every $\alpha>0$. In particular, $((1-\lambda)a,z)\in\cA$ for all $\lambda\in(0,1)$. Since $\cA$ is closed, letting $\lambda\to0$ we conclude $(a,z)\in\cA$.
\end{proof}

\medskip

For $\cA\subset\R^N$, we denote by $\cone(\cA)$ the smallest convex cone containing $\cA$ and by $\overline{\cone}(\cA)$ its closure. The next result shows that surplus-invariant acceptance sets in finite dimensions are essentially translates of the positive cone.

\medskip

\begin{proposition}
Set $w:=(v_1,\dots,v_K,0,\dots,0)$. Then
\begin{equation}
\overline{\cone}(\cA-w)=\R^K_+\times\R^{N-K}\,.
\end{equation}
\end{proposition}
\begin{proof}
Clearly $\cA\subset \{a\in\R^K \,; \ a_j\ge v_j, \ \forall j=1,\dots,K\}\times\R^{N-K}$. Hence, $\overline{\cone}(\cA-w)\subset\R^K_+\times\R^{N-K}$. Note that for each $j=1,\cdots,K$ we have $(v_j+m)e_j\in\cA$ for every $m\ge 0$, which implies that $\lambda(v_je_j-w)+\lambda me_j$ belongs to $\cone(\cA-w)$ for every $\lambda>0$ and $m\ge 0$. Choosing $m:=\frac{1}{\lambda}$ we see that $\lambda(v_je_j-w)+e_j$ also belongs to $\cone(\cA-w)$ for every $\lambda>0$ so that, letting $\lambda\to 0$, we obtain $e_j\in\overline{\cone}(\cA-w)$ for all $j=1,\dots,K$. Since, by Lemma~\ref{unbounded components}, the vectors $\pm e_j\in\cA$ for $j=K+1,\dots,N$ we obtain $\R^K_+\times\R^{N-K}\subset\overline{\cone}(\cA-w)$. This concludes the proof.
\end{proof}

\medskip

\begin{remark}
A corresponding result does not hold in the case where $(\Omega,\cF,\probp)$ is an infinite probability space. Indeed for any sequence $(A_n)$ of pairwise disjoint, measurable sets of positive probability define
\begin{equation}
\cA:=\{X\in L^\infty \,; \ X1_{A_n}\ge -n1_{A_n}, \ \forall \ n\in\N\}\,.
\end{equation}
The set $\cA$ is easily seen to be a convex, surplus-invariant acceptance set that is $\sigma(L^\infty,L^1)$ closed. However, there cannot exist $W\in L^\infty$ with $\cA-W\subset L^\infty_+,$ since this would imply that the elements in $\cA$ are uniformly bounded from below.
\end{remark}

%%%%%%%%%%%%%%%%%%%%%%%%%%%%%%%%%%%%%%%%%%%%%%%%%%%%%%%%%%%%%%%%%%%%%%%%%%%%%%%%%%%%%%%%%%%%%%%%%%%%%%%%%%%%%%%%%%%%%%%%%%%%%%%%%%%%%%%%%%%
\section{Positive, surplus-invariant risk measures}
%%%%%%%%%%%%%%%%%%%%%%%%%%%%%%%%%%%%%%%%%%%%%%%%%%%%%%%%%%%%%%%%%%%%%%%%%%%%%%%%%%%%%%%%%%%%%%%%%%%%%%%%%%%%%%%%%%%%%%%%%%%%%%%%%%%%%%%%%%%

So far we have focused our attention on surplus-invariant acceptance sets $\cA\subset L^\infty$ and surplus-invariant risk measures of the form $\rho_{\cA,S}$, with eligible asset $S$. In this last section we seek to clarify the relationship between these objects and the class of risk measures studied by Cont, Deguest \& He in~\cite{ContDeguestHe2013} and by Staum in~\cite{Staum2013}, focusing on those properties we deem most relevant from a capital adequacy perspective. To enable the comparison, we work with cash-additive risk measures, i.e. the eligible asset is taken to be $S=(1,1_\Omega)$. Recall that in this case we write $\rho_\cA$ instead of $\rho_{\cA,S}$.

\medskip

\begin{remark}
Let $\cA\subset L^\infty$ be an acceptance set, then $\rho_\cA$ is well known to be finitely valued and Lipschitz continuous. Moreover, $\rho_\cA=\rho_{\overline{\cA}}$ and $\overline{\cA}=\{X\in L^\infty \,; \ \rho_\cA(X)\le 0\}$. In the sequel, we use these facts without further reference.
\end{remark}

\medskip

\begin{definition}
A risk measure $\rho$ is called a {\em positive risk measure} if it is normalized, i.e. $\rho(0)=0$, and if it takes only positive real values.
\end{definition}

\medskip

\begin{remark}
\label{remark on staum cont definitions}
\begin{enumerate}[(i)]
\item Positive, surplus-invariant risk measures coincide with the {\em shortfall} risk measures introduced in~\cite{Staum2013}. This is because, for a positive risk measure, surplus invariance is equivalent to
\begin{equation}
\label{staum xs invariance}
\rho(X)=\rho(-X^-) \ \ \ \mbox{for every} \ X\in L^\infty\,.
\end{equation}
\item The {\em loss-based} risk measures introduced in~\cite{ContDeguestHe2013} are positive, surplus-invariant risk measures $\rho$ satisfying the {\em cash-loss} property: $\rho(-\alpha)=\alpha$ for every $\alpha\ge 0$.
\end{enumerate}
\end{remark}

%%%%%%%%%%%%%%%%%%%%%%%%%%%%%%%%%%%%%%%%%%%%%%%

\subsection{Truncations of risk measures of the form $\rho_\cA$}

No risk measure of the form $\rho_\cA$ can be positive, since cash additivity implies that $\rho_\cA$ is unbounded from below.
There are, however, two natural ways of constructing positive risk measures from $\rho_\cA$. All but one of the examples considered in~\cite{ContDeguestHe2013} and~\cite{Staum2013}, which we recall below, rely on one of these two construction principles.

\medskip

\begin{example}
\label{upper and lower trunctations}
Let $\cA\subset L^\infty$ be an acceptance set.
\begin{enumerate}
  \item  The map $\rho_\cA^\wedge:L^\infty\to[0,\infty)$ is defined by setting
\begin{equation}
\rho_\cA^\wedge(X):=\rho_\cA(X\wedge 0)=\rho_\cA(-X^-) \ \ \ \mbox{for} \ X\in L^\infty\,.
\end{equation}
If $\rho_\cA(0)=0$, then $\rho_\cA^\wedge$ is a positive, surplus-invariant risk measure. This risk measure was called {\em loss-based version} of $\rho_\cA$ in~\cite{ContDeguestHe2013} and {\em excess-invariant counterpart} of $\rho_\cA$ in~\cite{Staum2013}.

\item The map $\rho_\cA^\vee:L^\infty\to[0,\infty)$ is defined by setting
\begin{equation}
\rho_\cA^\vee(X):=\rho_\cA(X)\vee0=\max\{\rho_\cA(X),0\} \ \ \ \mbox{for} \ X\in L^\infty\,.
\end{equation}
If $\rho_\cA(0)\le 0$, then $\rho_\cA^\vee$ is a positive risk measure that is, however, not always surplus invariant. In Proposition~\ref{wedge=vee} below we characterize when this is the case.
\end{enumerate}
\end{example}

\medskip

\begin{example}
\label{cont-staum examples}
The following list covers the main examples of risk measures $\rho$ treated in~\cite{ContDeguestHe2013} and in~\cite{Staum2013}. Note that all these risk measures can be written as $\rho_\cA^\wedge$ where $\cA\subset L^\infty$ is the (not necessarily surplus-invariant) acceptance set $\cA(\rho)$. For more details we refer to~\cite{ContDeguestHe2013} and~\cite{Staum2013}.
\begin{enumerate}
\item (Scenario-based margin requirements,~\cite{ContDeguestHe2013}) The \textit{scenario-based margin requirement}, or \textit{$\SPAN$ risk measure}, associated to $A\in\cF$ is the positive, surplus-invariant risk measure defined by
\begin{equation}
\rho(X):=\esssup(1_AX^-)\,,
\end{equation}
where $\esssup(X)$ denotes the essential supremum of a random variable $X\in L^\infty$.

\item (Put option premium,~\cite{ContDeguestHe2013}) The \textit{put option premium} is the positive, surplus-invariant risk measure defined by
\begin{equation}
\rho(X):=\E[X^-] \ \ \ \mbox{for} \ X\in L^\infty\,.
\end{equation}
This risk measure was studied by Jarrow in \cite{Jarrow2002}.

\item (Spectral loss measure,~\cite{ContDeguestHe2013}) Let $\varphi:(0,1)\to[0,\infty)$ be a decreasing function with $\int^{1}_{0}\varphi(x)dx=1$. The \textit{spectral loss measure} associated to $\varphi$ is the positive, surplus-invariant risk measure defined by
\begin{eqnarray}
\label{def spectral rm}
\rho(X)&:=&\int^{1}_{0}(\VaR_\beta(X)\vee0)\,\varphi(\beta)\,d\beta \\
&=& \int^{1}_{0}\VaR_\beta(-X^-)\,\varphi(\beta)\,d\beta \ \ \ \mbox{for} \ X\in L^\infty\,.
\end{eqnarray}
The corresponding nontruncated version was introduced in Acerbi~\cite{Acerbi2002}. Note that the equality in \eqref{def spectral rm} is due to the surplus invariance of Value-at-Risk.
In particular, setting $\varphi(\beta):=\frac{1}{\alpha}1_{(0,\alpha)}(\beta)$ for a given $\alpha\in(0,1)$ we obtain
\begin{equation}
\rho(X):=\TVaR_\alpha(-X^-) \ \ \ \mbox{for} \ X\in L^\infty\,,
\end{equation}
which is referred to as \textit{expected tail-loss} in~\cite{ContDeguestHe2013}.

\item (Shortfall risk,~\cite{Staum2013}) Let $u:\R\to\R$ be an increasing function such that $u(0)=0$. The \textit{shortfall risk measure} is the positive, surplus-invariant risk measure defined by
\begin{equation}
\rho(X):=-\E[u(-X^-)] \ \ \ \mbox{for} \ X\in L^\infty\,.
\end{equation}
This type of risk measures was introduced by F\"{o}llmer and Schied in~\cite{FoellmerSchied2002}.
\end{enumerate}
\end{example}

\medskip

The following example in~\cite{ContDeguestHe2013} is of interest because it displays a positive risk measure that cannot be written in the form $\rho_\cA^\wedge$ for any acceptance set $\cA\subset L^\infty$.

\medskip

\begin{example}[Loss certainty equivalent,~\cite{ContDeguestHe2013}]
\label{example: loss certainty equivalent}
Let $u:[0,\infty)\to\R$ be strictly increasing and strictly convex. The \textit{loss certainty equivalent} is the positive, surplus-invariant risk measure defined by
\begin{equation}
\rho(X):=u^{-1}(\E[u(X^-)]) \ \ \ \mbox{for} \ X\in L^\infty\,.
\end{equation}
As explained in~\cite{ContDeguestHe2013}, the loss certainty equivalent cannot be typically written in the form $\rho_\cA^\wedge$ for a suitable acceptance set.
\end{example}

\medskip

We now provide a characterization of when risk measures of the form $\rho_\cA^\vee$ are surplus invariant, establishing a link with surplus-invariant acceptance sets.

\medskip

\begin{proposition}
\label{wedge=vee}
Let $\cA\subset L^\infty$ be an acceptance set. The following statements are equivalent:
\begin{enumerate}[(a)]
\item $\rho_\cA^\vee$ is surplus invariant;
\item $\overline{\cA}$ is surplus invariant.
\end{enumerate}
If $\rho_\cA(0)=0$, then \textit{(a)} and \textit{(b)} above are equivalent to
 \begin{enumerate}
\item[(c)] $\rho_\cA^\vee=\rho_\cA^\wedge$.
\end{enumerate}
\end{proposition}
\begin{proof}
Assume \textit{(a)} holds. For any $X\in\overline{\cA}$ we have $\rho_{\cA}(X)\le 0$, hence $\rho_\cA^\vee(X)= 0$. Since $\rho_\cA^\vee$ is surplus invariant, this implies $\rho_\cA^\vee(-X^-)= 0$. Thus, $\rho_\cA(-X^-)\le 0$ so that $-X^-\in\overline{\cA}$ and \textit{(b)} holds by Lemma~\ref{lemma: xs invariance and indicators}. Hence, \textit{(a)} implies \textit{(b)}.

\smallskip

Conversely, assume \textit{(b)} holds. Then $\rho_{\overline{\cA}}$ is surplus invariant by Proposition~\ref{xs-invariant risk measure}, hence so is $\rho_\cA$ since $\rho_\cA=\rho_{\overline{\cA}}$. Now, take $X\in L^\infty$. If $\rho_\cA(X)\geq0$ then $\rho_\cA^\vee(-X^-)=\rho_\cA^\vee(X)$ follows since $\rho_\cA$ is surplus invariant. Otherwise, if $\rho_\cA(X)<0$ we have $X\in\cA$ so that $-X^-\in\cA$ by the surplus invariance of $\cA$. Hence, $\rho_\cA(-X^-)\leq0$ implying that $\rho_\cA^\vee(-X^-)=0=\rho_\cA^\vee(X)$. In conclusion, \textit{(b)} implies \textit{(a)}.

\smallskip

Assume now that $\rho_\cA(0)=0$ and that \textit{(b)} holds. Since $\rho_\cA=\rho_{\overline{\cA}}$, the risk measure $\rho_\cA$ is surplus invariant again by Proposition~\ref{xs-invariant risk measure}. If $\rho_\cA(X)\ge 0$ then $\rho_\cA^\vee(X)=\rho_\cA(X)=\rho_\cA(-X^-)=\rho_\cA^\wedge(X)$. On the other hand, if $\rho_\cA(X)<0$ then $X\in\cA$ and surplus invariance implies that $-X^-\in\overline{\cA}$. It follows that $0=\rho_\cA(0)\le \rho_\cA(-X^-)\le 0$ and, consequently, $\rho_\cA^\wedge(X)=0=\rho_\cA^\vee(X)$. Hence, \textit{(b)} implies \textit{(c)}. Finally, since $\rho_\cA^\wedge$ is always surplus invariant, \textit{(c)} implies \textit{(a)}, concluding the proof.
\end{proof}

%%%%%%%%%%%%%%%%%%%%%%%%%%%%%%%%%%%%%%%%%%%%%%%%

\subsubsection*{Cash additivity subject to positivity}

In this section we establish when a positive, surplus-invariant risk measure is of the form $\rho_\cA^\vee$ for some surplus-invariant acceptance set $\cA\subset L^\infty$. To this end, we recall the property of cash additivity subject to positivity introduced in~\cite{Staum2013}.

\medskip

\begin{definition}
A positive risk measure $\rho:L^\infty\to[0,\infty)$ is said to be {\em cash additive subject to positivity} if
\begin{equation}
\label{casp}
\rho(X+\alpha) = \rho(X) - \alpha
\end{equation}
holds for every $X\in L^\infty$ and $\alpha\in\R$ such that $\rho(X)>\alpha\vee0$.
\end{definition}

\medskip

\begin{proposition}
\label{characterization casp and truncated}
For a positive, surplus-invariant risk measure $\rho:L^{\infty}\to[0,\infty)$ the following statements are equivalent:
\begin{enumerate}[(a)]
  \item $\rho=\rho_\cA^\vee$ for some surplus-invariant acceptance set $\cA\subset L^\infty$;
  \item $\rho$ is cash additive subject to positivity.
\end{enumerate}
The acceptance set in \textit{(a)} can always be taken to be $\cA:=\cA(\rho)$.
\end{proposition}
\begin{proof}
The cash additivity of $\rho_\cA$ readily implies that $\rho_\cA^\vee$ is cash additive subject to positivity. Hence, \textit{(a)} implies \textit{(b)}. Now assume \textit{(b)} holds, and take $X\in L^\infty$. We claim that $\rho=\rho_{\cA(\rho)}^\vee$. If $\rho(X)=0$, then $\rho_{\cA(\rho)}(X)\leq0$ and thus $\rho(X)=0=\rho^\vee_{\cA(\rho)}(X)$. If $\rho(X)>0$, it follows by cash-additivity subject to positivity that for all $\e>0$
\begin{equation}
\label{repr casp: auxiliary}
0 \leq \rho(X+\rho(X)) \le \rho(X+\rho(X)-\e) = \rho(X)-\rho(X)+\e=\e\,.
\end{equation}
In particular, $X+\rho(X)-\e\notin\cA(\rho)$ so that $\rho_{\cA(\rho)}(X)\geq\rho(X)-\e$ for all $\e>0$. Letting $\e\to0$, we obtain that $\rho_{\cA(\rho)}(X)\geq\rho(X)$. Moreover, letting $\e\to0$ in~\eqref{repr casp: auxiliary} we also get $\rho(X+\rho(X))=0$ so that $\rho_{\cA(\rho)}(X)\leq\rho(X)$. This shows that $\rho(X)=\rho_{\cA(\rho)}(X)=\rho^\vee_{\cA(\rho)}(X)$, concluding the proof that \textit{(b)} implies \textit{(a)}.
\end{proof}

\medskip

\begin{remark}
The previous result holds for any positive risk measure $\rho:L^\infty\to[0,\infty)$ if we do not require that the acceptance set in \textit{(a)} is surplus invariant.
\end{remark}

%%%%%%%%%%%%%%%%%%%%%%%%%%%%%%%%%%%%%%%%%%%%%%%%

\subsubsection*{Cash-loss additivity}

The property of cash-loss additivity introduced in~\cite{ContDeguestHe2013} characterizes positive, surplus-invariant risk measures of the form $\rho_\cA^\wedge$.

\medskip

\begin{definition}
A positive risk measure $\rho:L^\infty\to[0,\infty)$ is said to be {\em cash-loss additive} if for every $X\leq0$ and $\alpha>0$ it satisfies
\begin{equation}
\rho(X-\alpha)=\rho(X)+\alpha\,.
\end{equation}
\end{definition}

\medskip

The following result is proved in Section 4.2 of~\cite{ContDeguestHe2013} under a convexity assumption that is, however, unnecessary for the proof.

\medskip

\begin{proposition}
\label{representation rho wedge}
Let $\rho:L^\infty\to[0,\infty)$ be a positive, surplus-invariant risk measure. The following are equivalent:
\begin{enumerate}[(a)]
  \item $\rho=\rho_\cA^\wedge$ for some surplus-invariant acceptance set $\cA\subset L^\infty$;
  \item $\rho$ is cash loss additive.
\end{enumerate}
The acceptance set in \textit{(a)} can always be taken to be $\cA:=\cA(\rho)$.
\end{proposition}

\begin{remark}
The acceptance set in \textit{(a)} always satisfies $\rho_\cA(0)=0$, which is necessary for $\rho_\cA^\wedge$ to be a positive risk measure.
\end{remark}
%%%%%%%%%%%%%%%%%%%%%%%%%%%%%%%%%%%%%%%%%%%%%%%%

\subsection{Capital adequacy: Acceptability and operational meaning}
\label{subsection: accept and oper meaning}

In this section we focus on two fundamental questions that need to be addressed when interpreting a positive risk measure $\rho$ as a capital requirement: What is the acceptability criterion implicit in $\rho$? What is the operational meaning of $\rho(X)$?

%%%%%%%%%%%%%%%%%%%%%%%%%%%%%%%%%%%%%%%%%%%%%%%%

\subsubsection*{Acceptability}

As already mentioned, when a positive risk measure $\rho$ is interpreted as a capital requirement, there is a natural capital adequacy test, or acceptability criterion, associated with it: Acceptable positions are precisely those positions that do not require any additional capital, i.e. those positions belonging to
\begin{equation}
\cA(\rho):=\{X\in L^\infty \,; \ \rho(X)=0\}\,.
\end{equation}

\medskip

The acceptability criteria induced by positive, surplus-invariant risk measures are always surplus invariant. This simple result establishes a key link between the type of risk measures considered in~\cite{ContDeguestHe2013} and~\cite{Staum2013} and the theory developed in this paper.

\medskip

\begin{lemma}
If $\rho:L^\infty\to[0,\infty)$ is a positive, surplus-invariant risk measure, then the acceptance set $\cA(\rho)$ is surplus invariant.
\end{lemma}
\begin{proof}
If $X\in\cA(\rho)$ then $\rho(-X^-)=\rho(X)=0$, showing that $-X^-\in\cA(\rho)$. Hence, $\cA(\rho)$ is surplus invariant by Lemma~\ref{lemma: xs invariance and indicators}.
\end{proof}

\medskip

The following result plays a critical role in understanding the implicit acceptability criterion associated with the positive, surplus-invariant risk measures considered in~\cite{ContDeguestHe2013} and~\cite{Staum2013}. Recall that a risk measure $\rho$ is called sensitive if $\rho(X)>0$ for all nonzero $X\in L^\infty$ such that $X\leq0$.

\medskip

\begin{proposition}
\label{triviality with sensitivity}
Let $\rho:L^\infty\to[0,\infty)$ be a positive, surplus-invariant risk measure. If $\rho$ is sensitive, then $\cA(\rho)=L^\infty_+$.
\end{proposition}
\begin{proof}
Clearly, we have $L^\infty_+\subset\cA(\rho)$. Indeed, if $X\in L^\infty_+$ then $0\leq\rho(X)\leq\rho(0)=0$, showing that $X\in\cA(\rho)$. To prove the converse inclusion, take $X\in\cA(\rho)$ and assume $\probp(X<0)>0$ so that $X^-$ is nonzero. Since $\rho$ is sensitive, it follows that $\rho(X)=\rho(-X^-)>0$ and, therefore, $X\notin\cA(\rho)$.
\end{proof}

\medskip

\begin{remark}
The put option premium and the spectral loss measure in Example~\ref{cont-staum examples}, as well as the loss certainty equivalent in Example~\ref{example: loss certainty equivalent}, are sensitive. As a result, the preceding proposition implies that the corresponding acceptance sets collapse down to the positive cone. This is also true for shortfall risk introduced in Example~\ref{cont-staum examples} whenever $u$ is strictly increasing.
\end{remark}

\medskip

The case of positive, surplus-invariant risk measures of the form $\rho_\cA^\vee$ and $\rho_\cA^\wedge$ is particularly interesting. Here, the acceptance set $\cA$ is specified in advance. It is therefore natural to investigate whether the original acceptability criterion is preserved when making the transition from $\rho_\cA$ to $\rho_\cA^\vee$ and $\rho_\cA^\wedge$, respectively. When passing from $\rho_\cA$ to $\rho_\cA^\vee$ the acceptance set remains the same up to closure. However, as Proposition~\ref{truncated risk measures and acceptability} shows, in the transition to $\rho_\cA^\wedge$ the underlying acceptability criterion is maintained up to closure if and only if $\overline{\cA}$ is surplus invariant. In this case $\rho_\cA^\wedge$ and $\rho_\cA^\vee$ coincide by Proposition~\ref{wedge=vee}. If $\overline{\cA}$ is not surplus invariant, the acceptability criterion may be drastically altered as shown by Example~\ref{example drastic change} below. As a consequence, since the acceptance set $\cA$ already embodies the attitude of regulators towards risk, it seems that risk measures of the form $\rho_\cA^\wedge$ are less interesting from a capital adequacy perspective than those of the form $\rho_\cA^\vee$.

\medskip

\begin{proposition}
\label{truncated risk measures and acceptability}
Let $\cA\subset L^\infty$ be an acceptance set. The following statements hold:
\begin{enumerate}[(i)]
 \item $\cA(\rho_\cA^\vee)=\overline{\cA}$;
 \item If $\rho_\cA(0)=0$, then $\cA(\rho_\cA^\wedge)\subset\overline{\cA}$. Moreover, $\cA(\rho_\cA^\wedge)=\overline{\cA}$ if and only if $\overline{\cA}$ is surplus invariant.
 \end{enumerate}
\end{proposition}
\begin{proof}
\textit{(i)} This follows immediately since
\begin{equation}
\label{auxiliary, consequence of cont of cash add rm}
\overline{\cA} = \{X\in L^\infty \,; \ \rho_\cA(X)\leq0\} = \{X\in L^\infty \,; \ \rho_\cA^\vee(X)=0\} = \cA(\rho_\cA^\vee)\,.
\end{equation}
\textit{(ii)} Since $\cA(\rho_\cA^\wedge)=\{X\in L^\infty \,; \ \rho_\cA(-X^-)=0\}$, we immediately obtain $\cA(\rho_\cA^\wedge)\subset\overline{\cA}$. Thus, we only need to characterize when equality holds. The ``if'' part follows from part \textit{(i)} since in this case $\rho_\cA^\vee=\rho_\cA^\wedge$ by Proposition~\ref{wedge=vee}. To prove the ``only if'' part assume $\overline{\cA}\subset\cA(\rho_\cA^\wedge)$. Hence, for any $X\in\overline{\cA}$ we have $\rho_\cA(-X^-)=\rho_\cA^\wedge(X)=0$. This implies that $-X^-\in\overline{\cA}$, so that $\overline{\cA}$ is surplus invariant by Lemma~\ref{lemma: xs invariance and indicators}.
\end{proof}

\medskip

\begin{example}
\label{example drastic change}
Let $\cA\subset L^\infty$ be an acceptance set.
\begin{enumerate}
 \item Assume $\cA$ is closed and sensitive, and $\rho_\cA(0)=0$. Then $\cA(\rho_\cA^\wedge)=L^\infty_+$. This follows from \ref{triviality with sensitivity} since in that case $\rho^\wedge_\cA$ is sensitive.
 \item Assume $\cA$ is $\sigma(L^\infty,L^1)$ closed and coherent, and $\rho_\cA(0)=0$. Then $\cA(\rho_\cA^\wedge)=\SPAN(A)$ for some $A\in\cF$. This is a direct consequence of Theorem \ref{theorem coherent}, because the acceptance set $\cA(\rho_\cA^\wedge)$ is easily seen to be itself $\sigma(L^\infty,L^1)$ closed and coherent.
\end{enumerate}
\end{example}

%%%%%%%%%%%%%%%%%%%%%%%%%%%%%%%%%%%%%%%%%%%%

\subsubsection*{Operational meaning}

Let $\rho$ be a positive risk measure. In a capital adequacy context the quantity $\rho(X)$ is interpreted as the cost of making $X$ acceptable. Therefore, to provide $\rho$ with an operational meaning we need to specify of which action that makes $X$ acceptable is $\rho(X)$ the cost. In the absence of such a specification, $\rho(X)$ cannot be interpreted as a capital requirement in any practical way.

\medskip

Consider an acceptance set $\cA\subset L^\infty$. As already mentioned, by its very definition the cash-additive risk measure $\rho_\cA$ has a clear operational meaning: For any $X\in L^\infty$, the quantity $\rho_\cA(X)$ is the ``minimum'' amount of capital that has to be raised and added to $X$ to reach acceptability (up to closure of $\cA$). In this sense, the identity
\begin{equation}
\label{cash additivity identity}
\rho_\cA(X+\rho_\cA(X))=0 \ \ \ \mbox{for all} \ X\in L^\infty
\end{equation}
may be viewed as the key operational property. The following property -- which is also highlighted in Section 1.1 of~\cite{ContDeguestHe2013} -- is the natural generalization of~\eqref{cash additivity identity} to the context of positive risk measures.

\medskip

\begin{definition}
A positive risk measure $\rho:L^\infty\to[0,\infty)$ is called {\em cash compatible} if
\begin{equation}
\label{self balancing formula}
\rho(X+\rho(X)) = 0 \ \ \ \mbox{for all} \ X\in L^\infty\,.
\end{equation}
\end{definition}

\medskip

\begin{remark}
Cash compatibility has the following operational interpretation: If $\rho:L^\infty\to[0,\infty)$ is a cash-compatible positive risk measure and $X\in L^\infty$ an arbitrary capital position, then we can make $X$ acceptable by adding $\rho(X)$ to it. Hence, the risk measure, understood as a capital requirement, is compatible with the management action of adding cash to a position to make it acceptable.
\end{remark}

\medskip

\begin{remark}
\label{sb truncated rm}
Let $\cA\subset L^\infty$ be an acceptance set.
\begin{enumerate}[(i)]
\item Since $\rho_\cA$ is cash compatible, it follows that $\rho_\cA^\vee$ always satisfies~\ref{self balancing formula}. As a result, Proposition \ref{characterization casp and truncated} implies that every positive risk measure that is cash additive subject to positivity is automatically cash compatible.
\item The risk measure $\rho_\cA^\wedge$ does not always satisfy~\ref{self balancing formula}. For instance, the put option premium defined in Example~\ref{cont-staum examples} is easily seen not to have this property.
\end{enumerate}
\end{remark}

\medskip

We recall the concept of cash subadditivity introduced by El Karoui and Ravanelli in \cite{ElKarouiRavanelli2009}. For a discussion of this axiom in the context of risk measures of the form $\rho_{\cA,S}$ we refer to~\cite{FarkasKochMunari2013b}.

\medskip

\begin{definition}
A risk measure $\rho:L^\infty\to\overline{\R}$ is called {\em cash subadditive} if
\begin{equation}
\rho(X+\alpha) \geq \rho(X)-\alpha \ \ \ \mbox{for all} \ \alpha>0\,.
\end{equation}
\end{definition}

\medskip

\begin{remark}
\label{remark cash sub}
Let $\cA\subset L^\infty$ be an acceptance set.
\begin{enumerate}[(i)]
  \item The risk measure $\rho_\cA^\vee$ is always cash subadditive, as also shown in Proposition 4.3 in \cite{Staum2013}. As a consequence of Proposition \ref{characterization casp and truncated}, every positive, surplus-invariant risk measure is cash subadditive whenever it is cash additive subject to positivity.
  \item The risk measure $\rho_\cA^\wedge$ is always cash subadditive. Indeed, for $X\in L^\infty$ and $\alpha>0$ we have $-(X+\alpha)^-\leq-X^-+\alpha$ implying
\begin{equation}
\rho_\cA^\wedge(X+\alpha) = \rho_\cA(-(X+\alpha)^-) \geq \rho_\cA(-X^-+\alpha) = \rho_\cA^\wedge(X)-\alpha\,.
\end{equation}
As a consequence of Proposition \ref{representation rho wedge}, every positive, surplus-invariant risk measures that is cash-loss additive is automatically cash subadditive.
  \item The loss-based risk measures considered in~\cite{ContDeguestHe2013} are cash subadditive whenever convex, as shown in Section 2.1 in~\cite{ContDeguestHe2013}.
\end{enumerate}
\end{remark}

\medskip

The following theorem shows that, under the operationally relevant assumption of being cash compatible, every positive, surplus-invariant risk measure that is cash subadditive can be represented in the form $\rho_\cA^\vee$ for some surplus-invariant acceptance set $\cA\subset L^\infty$.

\medskip

\begin{theorem}
\label{repr rho vee under self balancing}
Let $\rho:L^\infty\to [0,\infty)$ be a positive, surplus-invariant risk measure that is cash compatible. Then, the following statements are equivalent:
\begin{enumerate}[(a)]
\item $\rho=\rho_\cA^\vee$ for some surplus-invariant acceptance set $\cA\subset L^\infty$;
\item $\rho$ is cash subadditive.
\end{enumerate}
The acceptance set in \textit{(a)} can always be chosen to be $\cA(\rho)$.
\end{theorem}
\begin{proof}
We only need to prove that \textit{(b)} implies \textit{(a)}, as the converse implication follows from Remark \ref{remark cash sub}. To this end, take $X\in L^\infty$. We claim that $\rho=\rho_{\cA(\rho)}^\vee$. Since $\rho$ is cash compatible we have $\rho(X+\rho(X))=0$ so that $X+\rho(X)\in\cA(\rho)$ and $\rho_{\cA(\rho)}^\vee(X)\leq\rho(X)$. If $\rho(X)=0$ we immediately get $\rho_{\cA(\rho)}^\vee(X)=\rho(X)$, hence assume $\rho(X)>0$. Take an arbitrary $0<\alpha<\rho(X)$. By cash subadditivity we have $\rho(X+\alpha)\ge\rho(X)-\alpha>0$ so that $X+\alpha\notin\cA(\rho)$. As a result, $\rho_{\cA(\rho)}^\vee(X)\geq\alpha$ for all such $\alpha$, implying that $\rho_{\cA(\rho)}^\vee(X)\geq\rho(X)$. We conclude that \textit{(b)} holds for $\cA:=\cA(\rho)$.
\end{proof}

\medskip

An immediate consequence of the preceding theorem is that all convex, loss-based risk measures studied in~\cite{ContDeguestHe2013} can be represented in the form $\rho_\cA^\vee$ as soon as they are cash compatible. This highlights, once again, the relevance of risk measures of the form $\rho_{\cA}^\vee$ with respect to a surplus-invariant acceptance set $\cA\subset L^\infty$.

\medskip

\begin{corollary}
Let $\rho:L^\infty\to[0,\infty)$ be a convex loss-based risk measure. If $\rho$ is cash compatible, then $\rho=\rho_{\cA(\rho)}^\vee$.
\end{corollary}

%%%%%%%%%%%%%%%%%%%%%%%%%%%%%%%%%%%%%%%%%%

%%%%%%%%%%%%%%%%%%%%%%%%%%%%%%%%%%%%%%%%%%%%%%%%%%%%%%%%%%%%%%%%%%%%%%%%%%%%%%%%%%%%%%%%%%%%%%%%%%%%%%%%%%%%%%%%%%%%%%%%%%%%%%%%%%%%%%%%%
\bibliographystyle{amsplain}
%\nocite{*}
%\section*{References}

\end{document}